\documentclass[12pt]{article}

\usepackage{amsmath,amsthm,amsfonts,amssymb,amscd}
\usepackage{xspace}
\usepackage{algorithm}
\usepackage[noend]{algpseudocode}
\usepackage{algpseudocode}
\usepackage{authblk} 
\usepackage{url}
\usepackage{color}

\newtheorem{theorem}{Theorem}
\newtheorem{lemma}{Lemma}

\newtheorem{definition}{Definition}

\newcommand{\bigo}[1]{\ensuremath{{{O} \left( {#1} \right)}}}

\newcommand{\bigomega}[1]{\ensuremath{{\Omega \left( {#1} \right)}}}

\newcommand{\bigtheta}[1]{\ensuremath{{\Theta(#1)}}}
\newcommand{\eps}{\ensuremath{\varepsilon}}

\newcommand{\R}{\ensuremath{\mathbb{R}}}
\newcommand{\C}{\ensuremath{\mathbb{C}}}
\newcommand{\F}{\ensuremath{\mathbb{F}}}

\newcommand{\set}[1]{\ensuremath{\left\{ {#1} \right\}}}
\newcommand{\abs}[1]{\ensuremath{\left| {#1} \right|}}

\newcommand{\zo}{\ensuremath{\set{0,1}}}

\newcommand{\st}{\ensuremath{\;\big |\;}}

\newcommand{\event}{\ensuremath{\mathcal{E}}}

\newcommand{\poly}[1]{\ensuremath{\mathsf{poly}\left( {#1} \right)}}
\newcommand{\polylog}[1]{\ensuremath{\mathsf{polylog}\left( {#1} \right)}}

\newcommand{\mat}[1]{\ensuremath{\mathbf{#1}}}
\newcommand{\rank}[1]{\ensuremath{\mathsf{rank}\left({#1}\right)}}
\newcommand{\tr}[1]{\ensuremath{\mathsf{tr}\left({#1}\right)}}
\renewcommand{\vec}[1]{\ensuremath{\mathbf{#1}}}
\newcommand{\dist}[1]{\ensuremath{\mathcal{#1}}}
\newcommand{\D}{\ensuremath{\dist{D}}}
\newcommand{\G}{\ensuremath{\dist{G}}}

\newcommand{\diag}{\ensuremath{\mathsf{diag}}}

\newcommand{\dtv}[1]{\ensuremath{\mathsf{D_{TV}}\left(#1\right)}}

\newcommand{\eqdef}{\stackrel{\mbox{\tiny{def}}}{=}}

\newcommand{\disjoint}{\textsc{Disjointness}\xspace}

\begin{document}

\title{Querying a Matrix through Matrix-Vector Products\footnote{We want to thank Roman Vershynin and Yan Shuo Tan for the helpful comments. David Woodruff would like to thank the Chinese Academy of Sciences, as well as the Simons Institute for the Theory of Computing where part of this work was done. }}



\author[1,2]{Xiaoming Sun}
\author[3]{David P. Woodruff}
\author[1,4]{Guang Yang}
\author[1,2]{Jialin Zhang}
	
\affil[1]{Institute of Computing Technology, Chinese Academy of Sciences, Beijing, China}
\affil[2]{University of Chinese Academy of Sciences, Beijing, China}
\affil[3]{Carnegie Mellon University, US}
\affil[4]{Conflux, Beijing, China}

\date{}

\maketitle

\begin{abstract}

We consider algorithms with access to
an unknown matrix $\mat{M} \in \mathbb{F}^{n \times d}$ via {\it matrix-vector products},
namely, the algorithm chooses vectors $\vec{v}^1, \ldots, \vec{v}^q$,
and observes $\mat{M}\vec{v}^1, \ldots, \mat{M}\vec{v}^q$.
Here the $\vec{v}^i$ can be randomized as well as chosen adaptively as
a function of $\mat{M}\vec{v}^1, \ldots, \mat{M}\vec{v}^{i-1}$. Motivated by applications
of sketching in distributed computation, linear algebra, and streaming models,
as well as connections to areas such as communication complexity and property
testing, we initiate the study of the number $q$ of queries
needed to solve various fundamental problems.
We study problems in three broad categories, including linear algebra, statistics problems, and graph problems.
For example, we consider the number of
queries required
to approximate the rank, trace, maximum eigenvalue,
and norms of a matrix $\mat{M}$;
to compute the AND/OR/Parity of each column or row of $\mat{M}$,
to decide whether there are identical columns or rows in $\mat{M}$ or
whether $\mat{M}$ is symmetric, diagonal, or unitary;
or to compute whether a graph defined by $\mat{M}$ is connected or triangle-free.
We also show separations for algorithms that are allowed to obtain
matrix-vector products only by querying vectors on the right, versus algorithms
that can query vectors on both the left and the right. We also show separations depending
on the underlying field the matrix-vector product occurs in.
For graph problems, we show separations depending on the form of the matrix (bipartite adjacency versus signed edge-vertex
incidence matrix) to represent the graph.

Surprisingly, this
fundamental model does not appear to have been studied on its own, and we
believe a thorough investigation of problems in this model would be beneficial
to a number of different application areas.
\end{abstract}

\section{Introduction}
 Suppose there is an unknown matrix $\mat{M} \in \mathbb{F}^{n \times d}$ that you can only access via a sequence of {\it matrix-vector products} $\mat{M} \cdot \vec{v}^1, \ldots, \mat{M} \cdot \vec{v}^q$, where we call the vectors $\vec{v}^1, \ldots, \vec{v}^q$ the {\it query vectors}, which can be chosen in a randomized, possibly adaptive way. By adaptive, we mean that $\vec{v}^i$ can depend on $\vec{v}^1, \ldots, \vec{v}^{i-1}$ as well as $\mat{M}\vec{v}^1, \ldots, \mat{M}\vec{v}^{i-1}$. Here $\mathbb{F}$ is a field, and we study different fields for different applications. Suppose our goal is to determine if $\mat{M}$ satisfies a specific property $\mathcal{P}$, such as having approximately full rank, or for example whether $\mat{M}$ has two identical columns. A natural question is the following:
\begin{center}
{\bf Question 1:} How many queries $q$ are necessary to determine if $\mat{M}$ has property $\mathcal{P}$?
\end{center}
A number of well-studied problems are special cases of this question, i.e., compressed sensing or sparse recovery, for which $\mat{M} \in \mathbb{R}^{1 \times d}$ is an approximately $k$-sparse vector, and one would like a number $q$ of queries close to $k$. It is known that if the query sequence is non-adaptive, meaning $\vec{v}^1, \ldots, \vec{v}^q$ are chosen before making any queries, then $q = \Theta(k \log(n/k))$ is necessary and sufficient \cite{c06,BIPW10} to recover an approximately $k$-sparse vector\footnote{Here the goal is to output a vector $\mat{M}'$ for which $\|\mat{M}-\mat{M}'\|_2 \leq (1+\epsilon)\|\mat{M}-\mat{M}_k\|_2$, where $\mat{M}_k$ is the best $k$-sparse approximation to $\mat{M}$, and $\epsilon$ is a constant.}. However, if the queries can be adaptive, then $q = O(k \log \log n)$ queries suffice \cite{ipw11}, while there is a lower bound of $\Omega(k + \log \log n)$ \cite{pw13} (see also recent work \cite{nsw018,kp19}).

The above problem is representative of an emerging field called {\it linear sketching} which is the underlying technique behind a number of algorithmic advances the past two decades.  In this model one queries $\mat{M} \cdot \vec{v}^1, \ldots, \mat{M} \cdot \vec{v}^r$ for non-adaptive queries $\vec{v}^1, \ldots, \vec{v}^r$. For brevity we write this as $\mat{M} \cdot \mat{V}$, where $\mat{V} \in \mathbb{F}^{d \times r}$ has $i$-th column equal to $\vec{v}^i$. Linear sketching has played a central role in the development of streaming algorithms \cite{ams96}. Perhaps more surprisingly, linear sketches are also known to achieve the minimal space necessary {\it of any, possibly non-linear, algorithm} for processing dynamic data streams under certain general conditions \cite{lnw14,ahlw16,kmsy18}, which is an essential result for proving a number of lower bounds for approximating matchings in a stream \cite{k15,akly16}. Linear sketching has also led to the fastest known algorithms for problems in numerical linear algebra, such as least squares regression and low rank approximation; for a survey see \cite{w14}. Note that given $\mat{M} \cdot \mat{V}$ and $\mat{M}' \cdot \mat{V}$, by linearity one can compute $(\mat{M}+\mat{M}') \cdot \mat{V} = \mat{M} \cdot \mat{V} + \mat{M}' \cdot \mat{V}$. This basic versatility property allows for fast updates in a data stream and mergeability in environments such as MapReduce and other distributed models of computation.

Given the applications above, we consider Question 1 an important question to understand for many different properties $\mathcal{P}$ of interest, which we describe in more detail below. A central goal of this work is to answer Question 1 for such properties and to propose this be a natural model of study in its own right.

One notable difference with our model and a number of appications of linear sketching is that we will allow for adaptive query sequences. In fact, our upper bounds will be non-adaptive, and our nearly matching lower bounds for each problem we consider will hold even for adaptive query sequences. Our model is also related to property testing, where one tries to infer properties of a large unknown object by (possibly adaptively) sampling a sublinear number of locations of that object. We argue that linear queries are a natural extension of sampling locations of an object, and that this is a natural ``sampling model'' not only because of the desired properties of the distributed, linear algebra, and streaming applications above, but sometimes also for physical constraints, e.g., in compressed sensing, where optical devices naturally capture linear measurements.

From a theoretical standpoint, any property testing algorithm, i.e., one that samples $q$ entries of $\mat{M}$, can be implemented in our model with $q$ linear queries. However, our model gives the algorithm much more flexibility.
From a lower bound perspective, as in the case of property testing \cite{bbm12}, some of our lower bounds will be derived from communication complexity. However, not all of our bounds can be proved this way. For example, one notable result we show is an optimal lower bound on the number of queries needed to approximate the rank of $\mat{M} \in \mathbb{R}^{n \times n}$ up to a factor $t$ by randomized, possibly adaptive algorithms; we show that $\frac{n}{t}+1$ queries are necessary and sufficient. A natural alternative way to prove this would be to give part of the matrix to Alice, part of to Bob, and have the players exchange the $\mat{M}^L\vec{v}^i$ and $\mat{M}^R \vec{v}^i$, where $\mat{M} = \mat{M}^L + \mat{M}^R$ and $\mat{M}^L$ is Alice's part and $\mat{M}^R$ is Bob's part. Then, if the $2$-player randomized communication complexity of approximating the rank of $\mat{M}$ up to a factor of $t$ were known to be $\Omega(n^2/t)$, we would obtain a nearly-matching query lower bound of $\Omega(n/(t (b + \log n)))$, where $b$ is the number of bits needed to specify the entries of $\mat{M}$ and the queries. However, it is unknown what the $2$-player communication complexity of approximating the rank of $\mat{M}$ up to a factor $t$ is over $\mathbb{R}$! We are not aware of any lower bound better than $\Omega(1)$ for constant $t$ for this problem for adaptive queries. We note that for non-adaptive queries, there is an $\Omega(n^2)$ sketching lower bound over the reals given in \cite{lnw14b},
and an $\Omega(n^2/\log p)$ lower bound for finite fields (of size $p$) in \cite{akl17}. There is also a property testing lower bound in \cite{blwz19}, though such a lower bound makes additional assumptions on the input. Thus, our model gives a new lens to study this problem from, from which we are able to derive strong lower bounds for adaptive queries. Our techniques could be helpful for proving
lower bounds in existing models, such as two-party communication complexity.

Our model is also related to linear decision tree complexity, see, e.g., \cite{b92,klm18}, though such lower bounds typically involve just seeing a threshold applied to $\mat{M}\vec{v}^i$, and typically $\mat{M}$ is a vector. In our case, we observe the entire output vector $\mat{M}\vec{v}^i$.

An interesting twist in our model is that in our formulation above, we only allowed to query $\mat{M}$ via matrix-vector products {\it on the right}, i.e., of the form $\mat{M} \cdot \vec{v}^i$. One could ask if there are natural properties $\mathcal{P}$ of $\mat{M}$ for which the number $q_L$ of queries one would need to make if querying $\mat{M}$ via queries of the form $(\vec{u}^1)^T\mat{M}, (\vec{u}^2)^T\mat{M}, \ldots, (\vec{u}^{q_L})^T\mat{M}$ can be significantly smaller than the number $q_R$ of queries one would need to make if querying $\mat{M}$ via queries of the form $\mat{M}\vec{u}^1, \mat{M}\vec{u}^2, \ldots, \mat{M}\vec{u}^{q_R}$:
\begin{center}
{\bf Question 2:} Are there natural problems for which $q_L \ll q_R$?
\end{center}
We show that this is in fact the case, namely, if we can only multiply on the right, then it takes $\Omega(n/\log n)$ queries to determine if there is a {\it column} of a matrix $\mat{M} \in \zo^{n \times n}$ which is all $1$s. However, if we can multiply on the left, then the single query $(1, 1, \ldots, 1)$ can determine this.

We study a few problems around Question 2, which is motivated from several perspectives. First, matrices might be stored on computers in a specific encoding, e.g., a sparse row format, from which it may be much easier to multiply on the right than on the left. Also, in compressed sensing, it may be natural for physical reasons to obtain linear combinations of columns rather than rows.

Another important question is how the query complexity depends
on the {\it underlying field} for which matrix-vector products are performed.
Might it be that for a natural problem the query complexity if the
matrix-vector products are performed modulo $2$ is much higher than
if the matrix-vector products are performed over the reals?

\begin{center}
{\bf Question 3:} Is there a natural problem for which the query complexity
in our model over $\F[2]$ is much larger than that over the reals?
\end{center}

Yet another important application of this model is to querying graphs. A natural
question is which representation to use for the graph. For example,
a natural representation of a graph on $n$ vertices is through its adjacency
matrix $\mat{A} \in \{0,1\}^{n \times n}$, where $\mat{A}_{i,j} = 1$ if and only if $\{i,j\}$
occurs as an edge. A natural representation for a bipartite graph with $n$
vertices in each part could be an $n \times n$ matrix $\mat{A}$ where $\mat{A}_{i,j} = 1$
iff there is an edge from the $i$-th left vertex to the $j$-th right vertex.
Yet another representation could be the $\binom{n}{2}\times n$
edge-vertex incidence matrix, where the $\{i,j\}$-th row is either
$0$, or has exactly two ones, one in location $i$ and one in location $j$.
One often considers a signed edge-vertex incidence matrix, where one first
arbitrarily fixes an ordering on the vertices and then the
$\{i,j\}$-th entry has a $1$ in the $i$-th position and a $-1$ in the
$j$-th position if $i > j$, otherwise positions $i$ and $j$ are swapped.
Yet another possible representation of a graph is through its Laplacian.

\begin{center}
{\bf Question 4:} Do some natural representations of graphs admit much
more efficient query algorithms for certain problems than other natural representations?
\end{center}
We note that in the data stream model, where one sees a long sequence
of insertions and deletions to the edges of a graph, each of the matrix
representations above can be simulated and so they lead to the same
complexity. We will show, perhaps surprisingly, that in this model
there can be an exponential difference in the query complexity for
two different natural representations of a graph for the same problem.

We next get into the details of our results.
We would like to stress
that even basic problems in this model are not immediately obvious how
to tackle. As a puzzle for the reader, what is the query complexity of
determining if a matrix $\mat{M} \in \mathbb{F}^{n \times n}$ is symmetric
if one can only query vectors on the right? We will answer this later in
the paper.

\subsection{Formal Model and Our Results}
We now describe our model and results formally in terms of an oracle.
The oracle has a matrix $\mat{M} \in \mathbb{F}^{m \times n}$, for some underlying
field $\mathbb{F}$ that we specify in each application. We can only query this
matrix via matrix-vector products, i.e., we pick an arbitrary vector $\vec{x}$
and send it to the oracle, and the oracle will respond with a vector $\vec{y}=\mat{M}\cdot\vec{x}$.
We focus our attention when the queries only occur on the right.
Our goal is to approximate or test a number of properties of $\mat{M}$ with a minimal number of queries, i.e.,
to answer Question 1 for a large number of different application areas.

We study a number of problems as summarized in the Table 1.
We assume $\mat{M}$ is an $m\times n$ matrix and $\eps > 0$ is a parameter of the problem.
The bounds hold for constant probability algorithms.
In some problems, such as testing
whether the matrix is a diagonal matrix, we always assume $m=n$,
and in the graph testing problems we explicitly describe
how the graph is represented using $\mat{M}$.
Interestingly, we are able to prove very strong lower bounds for approximating the rank,
which as described above, are unknown to hold for randomized communication complexity.

Motivated be streaming and statistics questions,
we next study the query complexity of approximating the norm of each row of $\mat{M}$.
We also study
the computation of the majority or parity of each column or row of $\mat{M}$,
the AND/OR of each column or row of $\mat{M}$, or equivalently,
whether $\mat{M}$ has an all ones column or row,
whether $\mat{M}$ has two identical columns or rows, and whether
$\mat{M}$ contains an unusually large-normed row, i.e., a ``heavy hitter''.
Here we show there are natural problems, such as computing the parity of all columns, which can be solved with $1$ query if sketching on the left, but require $\bigomega{n}$ queries if sketching on the right, thus answering Question 2.
We also answer Question 3, observing for the natural problem of testing if a row
is all ones, a single deterministic
query suffices over the reals but over $\F[2]$ this deterministically requires $\Omega(n)$
queries.

\begin{table}
\caption{Our Results}
\centering
\small
\begin{tabular}{|l|l|}
\hline
  Problem & Query Complexity \\
  \hline
  \multicolumn{2}{|c|}{Linear Algebra Problems}\\
  \hline
  Approximate Rank (for any $p'>p$ & $p+1$ (Section~\ref{sec:rank}) \\
   distinguishing Rank$\leq p$ from Rank $p'$) & \\
  \hline
  Trace Estimation  & $\bigomega{{n}/{\log n}}$ (Section~\ref{sec:tra})\\
  \hline
  Symmetric Matrix / Diagonal Matrix & $\bigo{1}$
  (Section~\ref{sec:sym} and ~\ref{sec:diag})\\
  \hline
  Unitary Matrix & 1 (Section~\ref{sec:unitary})\\
  \hline
  Approximate Maximum Eigenvalue & $\Theta(\eps^{-0.5}{\log n})$ for adaptive queries, \\
  & $\bigtheta{n}$ for non-adaptive queries (Section~\ref{sec:maxeigen}) \\
  \hline
  \multicolumn{2}{|c|}{Streaming and Statistics Problems}\\
  \hline
  All Ones Column & $\bigtheta{n}$ over $\F[2]$, \\
  & $\bigomega{{n}/{\log n}}$ over $\R$
   (Section~\ref{sec:all1}) \\
  \hline
  Two Identical Columns & $\Omega(n/m)\ (m=\Omega(\log (n/\epsilon)))$ \\
  Two Identical Rows & $\bigo{\log m}$ (Section~\ref{sec:identi})\\
  \hline
  Approximate Row Norms / Heavy Hitters & $\bigo{{\eps^{-2}}{\log m}}$  (Section~\ref{sec:norms})\\
  \hline
  Majority of Columns  & $\Omega({n/\log n})$ over $\R$ \\
  Majority of Rows  & $\bigo{1}$ over $\R$ (Section~\ref{sec:majority})\\
  \hline
  Parity of Columns  & $\bigtheta{n}$  \\
  Parity of Rows & $\bigo{1}$ (Section~\ref{sec:parity})\\
  \hline
  \multicolumn{2}{|c|}{Graph Problems}\\
  \hline
  Connectivity given Bipartite Adjacency Matrix & $\bigomega{n/\log n}$ (Section~\ref{sec:connect})\\
  Connectivity given Signed Edge-Vertex Matrix& $\bigo{\polylog{n}}$ (\cite{KLMMS17},
  noted in Section~\ref{sec:connect})\\
  \hline
  Triangle Detection & $\bigomega{{n}/{\log n}}$
  (Section~\ref{sec:triangle})\\
  \hline
\end{tabular}
\end{table}

For graph problems, we first argue if the graph is presented as an $n \times n$
bipartite adjacency matrix $\mat{M}$, then it requires $\Omega(n/\log n)$ possibly adaptive queries
to determine if the graph is connected. In contrast, if the graph is presented as
an $n \times \binom{n}{2}$ signed vertex-edge incidence matrix, then $\polylog{n}$ non-adaptive queries
suffices. This answers Question 4, showing that the type of representation of the graph is critical in this model.
Motivated by a large body
of
recent work on triangle counting (see, e.g., \cite{e17} and the references
therein), we also give strong negative results
for this problem in our model, which as with all of our lower bounds
unless explicitly stated otherwise,
hold even for algorithms which perform adaptive queries.


\section{Preliminaries}

We use capital bold letters, e.g., $\mat{A,B,M}$, to denote matrices, and use lowercase bold letters, e.g., $\vec{x,y}$, to denote column vectors.
Sometimes we write a matrix as a list of column vectors in square brackets, e.g., $\mat{M} = [\vec{m}_1,\dots,\vec{m}_n]$.
We use calligraphic letters, e.g., $\D$, to denote probability distributions,
and use $\mat{M}\gets \D$ to denote that $\mat{M}$ is sampled from distribution $\D$.
In particular, we use $\G$ to denote a Gaussian distribution and $\mat{G}$ for a matrix whose entries are sampled from an independently and identically distributed (denoted as i.i.d. in the following) Gaussian distribution.

We call a matrix $\mat{M}$ i.i.d. Gaussian if each element is i.i.d. Gaussian. It is easy to check that if matrix $\mat{G}$ is a $p\times n$ i.i.d. Gaussian matrix, and $\mat{R}$ is an $n\times n$ rotation matrix, then $\mat{G}\times \mat{R}$ is still i.i.d. Gaussian, and has the same probability distribution of $\mat{G}$.

The \emph{total variation distance}, sometimes called the statistical distance, between two probability measures $P$ and $Q$ is defined as 
 \[
 \dtv{P,Q} \eqdef \sup_A\abs{P(A)-Q(A)}.
 \]



Let $\mat{X}$ be an $n\times m$ matrix with each row i.i.d. drawn from an $m$-variate normal distribution $N(0,\Sigma)$.
Then the distribution of the $m\times m$ random matrix $\mat{A}=\mat{X}^T\mat{X}$ is called the \emph{Wishart distribution} with $n$ degrees of freedom and covariance matrix $\Sigma$,
 denoted by $W_m(n,\Sigma)$.
%
The distribution of eigenvalues of $\mat{A}$ is characterized in the following lemma.
\begin{lemma}[Corollary 3.2.19 in \cite{kent1984aspects}]\label{lem: wishart}
	If $\mat{A}$ is $W_m(n,\lambda I_m)$, with $n>m-1$, the joint density function of the eigenvalues $\mat{\Lambda}=(\lambda_1,\dots,\lambda_m)$ of $\mat{A}$ (in descending order) is
	\begin{align*}
		f(\mat{\Lambda}) = \frac{\pi^{m^2/2}}{(2\lambda)^{mn/2}\Gamma_m(m/2)\Gamma_n(n/2)} \exp\left(-\frac{1}{2\lambda}\sum_{i=1}^m \lambda_i \right)\prod_{i=1}^m \lambda_i^{(n-m-1)/2} \prod_{1\le i<j\le m} (\lambda_i-\lambda_j)
	\end{align*}
	In particular, for $\lambda=1$ and $n=m$, $\exists$ a constant  $Z_m$  independent from $\lambda_1,\dots,\lambda_m$, such that
	\begin{align*}
		f(\mat{\Lambda}) = \frac{1}{Z_m} \exp\left(-\frac{1}{2}\sum_{i=1}^m \lambda_i \right)\prod_{i=1}^m \lambda_i^{-1/2} \prod_{1\le i<j\le m} (\lambda_i-\lambda_j)
	\end{align*}
\end{lemma}


\section{Linear Algebra Problems}
In this section we present our lower bound for rank approximation in Section~\ref{sec:rank}. In the following, we provide our results about trace estimation in Section~\ref{sec:tra}, testing symmetric matrices in Section~\ref{sec:sym}, testing diagonal matrices in Section~\ref{sec:diag}, testing unitary matrices in Section~\ref{sec:unitary}, and approximating the maximum eigenvalue in Section~\ref{sec:maxeigen}.

\subsection{Lower Bound for Rank Approximation}\label{sec:rank}
In this section, we discuss how to approximate the rank of a given matrix $\mat{M}$ over the reals when the queries consist of right multiplication by vectors. A na\"ive algorithm to learn the rank is to pick random Gaussian query vectors non-adaptively.
In order to approximate the rank, that is, to distinguish whether $\rank{\mat{M}}\le p$ or $\rank{\mat{M}} \ge p+1$, this algorithm needs at least $p+1$ queries, and it is not hard to see that the algorithm succeeds with probability $1$.
Indeed, if $\mat{H} \in \mathbb{R}^{n \times (p+1)}$ is the random Gaussian query matrix, and $\mat{M}$ the unknown $n\times n$ matrix, then we can write $\mat{M}$ in its thin singular value decomposition as $\mat{M} = \mat{U} \mat{\Sigma} \mat{V}^T$, where $\mat{U}$ and $\mat{V}\in \mathbb{R}^{n\times k}$ have orthonormal columns, and $\mat{\Sigma} \in \mathbb{R}^{k\times k}$ has positive diagonal entries. Here $k=$rank$(M)$. We have that rank$(\mat{M} \cdot \mat{H}) = $ rank$(\mat{V}^T \mat{H})$, which by rotational invariance of the Gaussian distribution is the the same as the rank of a random Gaussian matrix, which will be the minimum of $p+1$ and the rank of $\mat{M}$ with probability $1$.

In the following, we will show that we cannot expect anything better. We will first show for non-adaptive queries, at least $p+1$ queries are necessary to learn the approximate rank. Then we generalize our results to adaptive queries. Our results hold for randomized algorithms by applying Yao's minimax principle.

\subsubsection{Non-Adaptive Query Protocols}

\begin{theorem}\label{thm:rank}
	Let constant $\eps>0$ be the error tolerance and let $\mat{M}$ be an $n\times n$ oracle matrix and suppose to start that we make non-adaptive queries.
	For integer $p < p'\le n$, at least $p+1$ queries are necessary to distinguish $\rank{\mat{M}}\le p$ from $\rank{\mat{M}} \ge p'$ with advantage $\ge \eps$.
\end{theorem}

\begin{proof}
	Given any algorithm distinguishing $\rank{\mat{M}}\le p$ from $\rank{\mat{M}}\ge p'$ for some $p'<n$,
	we can determine whether a $p'\times p'$ matrix $\mat{M}'$ has full rank $p'$ or $\rank{\mat{M}'}\le p$, by padding $\mat{M}'$ to an $n\times n$ matrix $\mat{M}$.
	Therefore in what follows it suffices to prove the lower bound for two $n\times n$ matrices $\mat{M}_1$ and $\mat{M}_2$ where $\rank{\mat{M}_1}\le p$ and $\rank{\mat{M}_2}=n$:
	\begin{enumerate}
		\item $\mat{M}_1=\mat{U}\times \mat{G}^T$;

		\item $\mat{M}_2=\mat{U}\times \mat{G}^T + \frac{1}{Z(n)}\cdot \mat{U}^{\bot} \times \mat{H}^T$.
	\end{enumerate}
	Here $\mat{U}$ has $p$ columns and $\mat{U}^{\bot}$  has $(n-p)$ columns such that $\left[\mat{U},\mat{U}^{\bot}\right]$ forms an $n\times n$ random orthonormal basis,
	$\mat{G}^T$ and $\mat{H}^T$ are $p\times n$ and $(n-p)\times n$ matrices whose entries are sampled i.i.d. from the standard Gaussian distribution,
	and $Z(n)$ is a function in $n$ which will be specified later.
	It immediately follows that $\rank{\mat{M}_1}\le p$ and $\rank{\mat{M}_2}=n$ with overwhelmingly high probability.
	Then we assume $\rank{\mat{M}_2}=n$  and discuss the query lower bound for distinguishing $\mat{M}_1$ from $\mat{M}_2$.

	Given  $\mat{M}\in\set{\mat{M}_1,\mat{M}_2}$, without loss of generality we denote the $q$ non-adaptive queries with an $n\times q$ orthonormal\footnote{Non-orthonormal queries can be made orthonormal using a change of basis in post-processing.} matrix $\mat{V}=[\vec{v}_1,\dots,\vec{v}_q]$, where $q\le p$ and each $n\times 1$ column vector $\vec{v}_i$ is a query to the oracle of matrix $\mat{M}$ which gets response $\mat{M}\cdot\vec{v}_i$, for $i\in[q]$.
	Then, it suffices to show that the following two distributions are hard to distinguish:
	\begin{enumerate}
		\item $\mat{M}_1\times \mat{V} \equiv \mat{U} \mat{W}$, where $\mat{W}=\mat{G}^T\mat{V}$;

		\item $\mat{M}_2\times \mat{V} \equiv \mat{U} \mat{W}+ \frac{1}{Z(n)}\cdot \mat{U}^{\bot} \mat{W}'$, where $\mat{W}'=\mat{H}^T\mat{V}$.
	\end{enumerate}

	Note that $\left[\mat{U},\mat{U}^{\bot}\right]$ is orthonormal, and hence $\mat{U}^T \mat{U}=\mat{I}_p, \left(\mat{U}^{\bot}\right)^T \mat{U}^{\bot}=\mat{I}_{n-p}$, $\mat{U}^T \mat{U}^{\bot}=\mat{0}_{p\times(n-p)}$.
	We introduce Lemma~\ref{lemma:Dtv} to eliminate $\mat{U},\mat{U}^{\bot}$ in the representation of $\mat{M}\times\mat{V}$.

	\begin{lemma}\label{lemma:Dtv}
		For $\mat{M}_1,\mat{M}_2$ and $\mat{V}$ defined as above,
		\[\dtv{\mat{M}_1\mat{V},\mat{M}_2\mat{V}} = \dtv{\left(\mat{M}_1\mat{V}\right)^T\mat{M}_1\mat{V}, \left(\mat{M}_2\mat{V}\right)^T\mat{M}_2\mat{V}}\]
	\end{lemma}

	\begin{proof}
		The direction $\dtv{\mat{M}_1\mat{V},\mat{M}_2\mat{V}} \ge \dtv{\left(\mat{M}_1\mat{V}\right)^T\mat{M}_1\mat{V}, \left(\mat{M}_2\mat{V}\right)^T\mat{M}_2\mat{V}}$ is trivial by the data processing inequality (i.e., for every $\mat{X,Y}$ and function $f$, $\dtv{\mat{X},\mat{Y}}\ge \dtv{f(\mat{X}),f(\mat{Y})}$).
		In what follows we only prove the other direction.

		First we notice that for every fixed $n\times n$ orthonormal matrix $\mat{R}$ and for a random matrix $\mat{M}$ sampled as $\mat{M}_1$ or $\mat{M}_2$, the product $\mat{N}\eqdef \mat{R}\mat{M}$ follows exactly the same distribution of $\mat{M}$.
		Thus $\mat{N}\mat{V}$ and $\mat{M}\mat{V}$ are identically distributed.

		Then, from a random sample $\mat{V}^T\mat{M}^T\mat{M}\mat{V}$ we can find $\mat{M}'$ such that $\mat{V}^T\mat{M}^T\mat{M}\mat{V} = (\mat{M}')^T\mat{M}'$ and $\mat{M}' = \mat{S}\mat{M}\mat{V}$ for some orthonormal matrix $\mat{S}$ and orthonormal query matrix $\mat{V}$.
		Although $\mat{M}'$ is not necessarily the same as $\mat{M}\mat{V}$ because of $\mat{S}$, we have $\mat{R}\mat{M}' \sim \mat{N}\mat{V} \sim \mat{M}\mat{V}$ for a uniformly random orthonormal matrix $\mat{R}$.
		Thus we transform a random sample from  $\mat{V}^T\mat{M}^T\mat{M}\mat{V}$ into a sample from $\mat{M}\mat{V}$ via $\mat{R}\mat{M}'$, and hence, we have
		$\dtv{\mat{M}_1\mat{V},\mat{M}_2\mat{V}} \le \dtv{\left(\mat{M}_1\mat{V}\right)^T\mat{M}_1\mat{V}, \left(\mat{M}_2\mat{V}\right)^T\mat{M}_2\mat{V}}$.
	\end{proof}

	Using Lemma~\ref{lemma:Dtv}, it suffices to prove an upper bound for $\dtv{ \mat{\Lambda}, \mat{\Lambda}' }$ as follows:
	\begin{align*}
	&	\dtv{\mat{U} \mat{W}, \mat{U} \mat{W}+ \frac{\mat{U}^{\bot} \mat{W}'}{Z(n)} }
	 \\
=& \dtv{ \left(\mat{U} \mat{W}\right)^T(\mat{U} \mat{W}), \left(\mat{U} \mat{W}+ \frac{\mat{U}^{\bot} \mat{W}'}{Z(n)}\right)^T(\mat{U} \mat{W}+ \frac{\mat{U}^{\bot} \mat{W}'}{Z(n)})  }\\
		=& \dtv{ \mat{W}^T\mat{W}, \mat{W}^T\mat{W} + \frac{(\mat{W}')^T\mat{W}'}{Z^2(n)} }
		\le \dtv{ \mat{\Lambda}, \mat{\Lambda}' }
	\end{align*}
	where $\mat{\Lambda} =\diag(\lambda_1,\dots,\lambda_q), \mat{\Lambda}' =\diag(\lambda'_1,\dots,\lambda'_q)$ are diagonal matrices such that $\mat{W}^T\mat{W} = \mat{A}^T\mat{\Lambda}\mat{A} $ and  $\mat{W}^T\mat{W}+ \frac{(\mat{W}')^T\mat{W}'}{Z^2(n)} = \mat{B}^T\mat{\Lambda}'\mat{B}$ for orthonormal matrices $\mat{A}$ and $\mat{B}$.
	The inequality follows because any algorithm separating $\mat{W}^T\mat{W}$ from $\mat{W}^T\mat{W}+\frac{(\mat{W}')^T\mat{W}'}{Z^2(n)}$ implies a separation of $\mat{\Lambda}$ from $\mat{\Lambda}'$ with the same advantage, by multiplying by random orthonormal matrices.

	By Weyl's inequality \cite{Weyl,Weyl2}, for every $i\in [q]$,
	$\lambda'_i \in \left[\lambda_i- \|\mat{\Lambda}'-\mat{\Lambda}\|_2, \lambda_i + \|\mat{\Lambda}'-\mat{\Lambda}\|_2\right]$, and hence $\lambda'_i \in \left[\lambda_i- \bigo{\frac{ \|\mat{W'}\|_2 }{Z^2(n)}}, \lambda_i + \bigo{\frac{\|\mat{W'}\|_2}{Z^2(n)}} \right]$.
	Notice that $\mat{W'}$ is an $(n-p)\times q$ i.i.d. Gaussian matrix, and hence  $\|\mat{W'}\|^2_2$ is a chi-squared variable with $(n-p)q$ degrees of freedom, which is bounded by $\bigo{(n-p)q}$ with high probability
	(c.f. Example 2.12 in \cite{Wainwright}).
	Recalling that $q\le p$, in what follows we condition on the event $\lambda'_i \in \left[\lambda_i- \bigo{\frac{ np }{Z^2(n)}}, \lambda_i + \bigo{\frac{ np }{Z^2(n)}} \right]$.

	We then show the gaps between eigenvalues $\lambda_i$ are sufficiently large. Note that since $\mat{G}^T$ is  i.i.d. Gaussian and $\mat{V}$ is an orthonormal matrix, each row in $\mat{W}=\mat{G}^T\mat{V}$  is independently drawn from an $q$-variate normal distribution, thus the probability distribution of  $\mat{W}^T\mat{W}$ is a Wishart distribution $W_q(p,I_q)$.
	Let $q=p$ and $\lambda_1,\dots,\lambda_p$ be sorted in descending order. Then by Lemma~\ref{lem: wishart} the density function of $\mat{\Lambda}$ is:
	\begin{align}\label{equ:density}
		f(\mat{\Lambda}) = \frac{1}{Z_p}\exp\left(-\frac{1}{2}\sum_{i=1}^p \lambda_i \right)\prod_{i=1}^p \lambda_i^{-1/2} \prod_{1\le i<j\le p} (\lambda_i-\lambda_j)
	\end{align}

	Let $\event$ denote the event that $\lambda_p\ge\frac{0.01}{\sqrt{n}}$ and $\forall 1\le i<j\le p, \lambda_i-\lambda_j \ge \gamma = 2^{-\bigtheta{p^2\log p}}$.

	\begin{lemma}\label{lem:lower bound of event}
		For  $\mat{W}^T\mat{W}$ defined as above and sufficiently small $\gamma = 2^{-\bigtheta{p^2\log n}}$, $\Pr[\event] > 0.9$.
	\end{lemma}
	\begin{proof}
		By equation (2) in \cite{shen2001} we know that $\Pr[\sqrt{n}\lambda_p\ge y] = \exp\left( -(y^2/2+y)\right)$.
		Thus for $y=0.01$ and $\event_0\eqdef \set{\lambda_p\ge 0.01/\sqrt{n}}$ we get:
		\[\Pr\left[\event_0\right]=\Pr\left[\lambda_p\ge \frac{0.01}{\sqrt{n}} \right]  =\exp\left(-0.01005\right)> 0.99000033\]
		Also, we note that for every $i$, $\Pr\left[ |\lambda_i| \le 100n \right] \ge 1-2\exp(-32n)$, by setting $t=8\sqrt{n}$ in Corollary 5.35 of \cite{V10}.
		In what follows we condition on the event
 $\event'_0$ that $|\lambda_i|\le 100n$ for every $i\in [p]$.

		Then we consider the joint distribution $\mu$ of $\lambda_1,\dots,\lambda_p$ in $\mat{\Lambda}$.
		Let $\event_i\eqdef \set{\lambda_i-\lambda_{i+1}<\gamma}$ be the event that $\lambda_i$ and $\lambda_{i+1}$ has a gap smaller than $\gamma$. Thus $\event= \event_0 \land\left( \land_{i=1}^{p-1}\overline{\event_i} \right)$.
		To lower bound $\Pr[\event]$, we need to upper bound the probability of $\event_i$ for $1\le i\le p-1$.

		Let $f$ be the density function of $\mu$ as in (\ref{equ:density}),
		and let $\mathsf{Leb}(\cdot)$ be the Lebesgue measure in $n$ dimensions.
		Then for every $i$,
		\[ \Pr[\event_i \st \event'_0] =\mu\left( \lambda_i-\lambda_{i+1}<\gamma \right) \le \mathsf{Leb}\left( \lambda_i-\lambda_{i+1}<\gamma \right) \cdot |f|_\infty = \bigo{\gamma/n}  \cdot |f|_\infty \]
		%
		Note that conditioning on $\event_0$ such that $\lambda_p\ge 0.01/\sqrt{n}$, the density function $f$ is bounded as:
		\[|f|_\infty \le \bigo{\exp\left(-\frac{1}{2}\lambda_1\right) \left(100\sqrt{n}\right)^{p/2} \lambda_1^{p^2/2}}
		= 2^{\bigo{p^2\log n}}\]
		As a result, we get $\Pr\left[\event_i \land \event_0 \st \event'_0\right ] \le {\gamma}\cdot 2^{\bigo{p^2\log n}}$.

		Therefore, the probability of $\event$ is lower bounded for sufficiently small $\gamma = 2^{-\bigtheta{p^2\log n}}$,
		\begin{align*}
			\Pr\left[\event\right] \ge& \Pr\left[\event'_0\right]\cdot \Pr\left[\event_0\land\left(\land_{i=1}^{p-1}\overline{\event_i}\right) \st \event'_0\right]\\
			\ge& \Pr\left[ \event'_0\right]\cdot\left(\Pr\left[\event_0\st \event'_0\right]- \sum_{i=1}^{p-1}\Pr\left[\event_i \land \event_0 \st \event'_0\right]\right)\\
			 >& \left(1-2p\exp(-32n)\right)\cdot \left(0.99000033 -(p-1)\gamma \cdot 2^{\bigo{p^2\log n}}\right)>0.9
		\end{align*}
	\end{proof}

	Conditioned on event $\event$ and recalling that $\lambda'_i \in \left[\lambda_i- \bigo{\frac{np}{Z^2(n)}}, \lambda_i + \bigo{\frac{np}{Z^2(n)}} \right]$,
	the probability density of $\mat{\Lambda}'$ has only a negligible difference from that of $\mat{\Lambda}$, since the small disturbance of eigenvalues is dominated by the corresponding terms in $f(\mat{\Lambda})$.
	\begin{align*}
		\frac{f(\mat{\Lambda}')}{ f(\mat{\Lambda}) } = & \frac{\exp\left(-\frac{1}{2}\sum_{i=1}^p \lambda'_i \right)\prod_{i=1}^p {\lambda'_i}^{-1/2} \prod_{1\le i<j\le p} (\lambda'_i-\lambda'_j) }{\exp\left(-\frac{1}{2}\sum_{i=1}^p \lambda_i \right)\prod_{i=1}^p \lambda_i^{-1/2} \prod_{1\le i<j\le p} (\lambda_i-\lambda_j)}\\
		\le & \exp\left(\frac{p\cdot np}{Z^2(n)}\right)  \left(\frac{\lambda_p-\frac{np}{Z^2(n)}}{\lambda_p}\right)^{-p/2}\prod_{1\le i<j\le p} \frac{\lambda_i-\lambda_j+\frac{2np}{Z^2(n)}}{\lambda_i-\lambda_j}\\
		\le &  \exp\left(\frac{np^2}{Z^2(n)}\right) \cdot\left(1+\frac{np}{\lambda_p \cdot Z^2(n)} \right)^p \left(1+\frac{2np}{Z^2(n) \cdot \min_{i\ne j}|\lambda_i-\lambda_j|}\right)^{p(p-1)/2}\\
		\le & \exp\left(\frac{np^2}{Z^2(n)}\right) \cdot\left(1+\frac{100\sqrt{n}\cdot np}{Z^2(n)} \right)^p
		\left(1+\frac{2np} {Z^2(n)\cdot\gamma }\right)^{p(p-1)/2} =1+\bigo{\frac{np^3 \gamma^{-1}}{Z^2(n)}}
	\end{align*}
	Similarly we can prove ${f(\mat{\Lambda}')}/{ f(\mat{\Lambda}) }  \ge 1-\bigo{{np^3 \gamma^{-1}}/{Z^2(n)}}$.
	Thus the total variation distance between $\mat{\Lambda}$ and  $\mat{\Lambda}'$ conditioned on $\event$ is $\dtv{ \mat{\Lambda}, \mat{\Lambda}'\st \event}\le \bigo{{np^3 \gamma^{-1}}/{Z^2(n)}}= \bigo{{1}/{n^2}}$ for sufficiently large $Z(n)\ge (np)^{1.5}\gamma^{-0.5} =2^{\bigtheta{p^2\log n}}$.
	Thus, for sufficiently large $n$, we have:
	\begin{align*}
		\dtv{ \mat{\Lambda}, \mat{\Lambda}'}\le \Pr[\overline{\event}] + \Pr[\event] \cdot \dtv{ \mat{\Lambda}, \mat{\Lambda}'\st \event} \le 0.1+\bigo{1/{n^2}} < 0.11
	\end{align*}

	Therefore, with as many as $q=p$ non-adaptive queries to the oracle matrix $\mat{M}$, the two distributions $\mat{M}_1$ and $\mat{M}_2$ cannot be distinguished with advantage greater than $0.11$.
	At least $p+1$ queries are necessary to distinguish those two matrices $\mat{M}_1$ and $\mat{M}_2$ of rank $\le p$ and rank $n$, respectively.

	Indeed, the above argument holds for every constant advantage $\eps$ if $y=\eps/3$, $t>\sqrt{12n/\eps}$, and $\gamma$ is sufficiently small in the proof of Lemma~\ref{lem:lower bound of event}, and letting $Z(n)$ be sufficiently large.	
\end{proof}


\subsubsection{Equivalence Between Adaptive and Non-Adaptive Protocols}

Now, we consider the adaptive query matrix $\mat{V}=[\vec{v}_1,\dots,\vec{v}_q]$ where $\vec{v}_i$ is the $i$-th query vector. Without loss of generality, we can assume that $\forall i, \vec{v}_i$ is a unit vector and it is orthogonal to query vectors $\vec{v}_1,\cdots, \vec{v}_{i-1}$. This gives us the following formal definition of an adaptive query protocol.

\begin{definition}
For a target matrix $\mat{M}$, an adaptive query protocol $P$ will output a sequence of query vectors $\vec{v}_1,\vec{v}_2,\cdots$. It is called a \textit{normalized adaptive protocol} if for any $i$, the query vector $\vec{v}_i$ output by $P$ satisfies
\begin{enumerate}
  \item $\vec{v}_i$ is a unit vector;
  \item $\vec{v}_i$ is orthogonal to the vectors $\vec{v}_1,\cdots, \vec{v}_{i-1}$;
  \item $\vec{v}_i$ is deterministically determined by $\mat{M} \times [\vec{v}_1,\dots,\vec{v}_{i-1}]$.
\end{enumerate}
\end{definition}

Let $P^{std}$ be a standard protocol which outputs $\vec{e}_1, \vec{e}_2, \cdots$ where $\vec{e}_i$ is the $i$-th standard basis vector. We then show that adaptivity is unnecessary by proving that $P^{std}$ has the same power as any normalized adaptive protocol to distinguish the matrix $\mat{M}_1$ and $\mat{M}_2$ defined in the previous section.

More formally, we show the following lemma for matrix $\mat{M}_2$:

\begin{lemma}
Fix any $n\times n$ orthogonal matrix $\left[\mat{U},\mat{U}^{\bot}\right]$ and any normalized adaptive protocol $P$. Consider $\mat{M}_2=\mat{U}\times \mat{G}^T + \frac{1}{\poly{n}}\cdot \mat{U}^{\bot} \times \mat{H}^T$ where $\mat{G}^T$ be a $p\times n$ i.i.d. Gaussian matrix, and $\mat{H}^T$ be a $(n-p)\times n$ i.i.d. Gaussian matrix. Let matrix $\mat{V}=[\vec{v}_1,\dots,\vec{v}_q]$ and $\mat{V}^{std}=[\vec{e}_1, \cdots, \vec{e}_q]$ be the query matrix output by protocol $P$ and $P^{std}$, correspondingly. We have the matrix $\mat{M}_2\mat{V}$ has the same distribution as $\mat{M}_2\mat{V}^{std}$.

\end{lemma}
\proof
Since the matrix $\mat{G}^T\cdot \mat{V}^{std}$ and $\mat{H}^T\cdot \mat{V}^{std}$ is i.i.d. Gaussian, that is, every element in two matrices is from standard Gaussian distribution and independent of each other, it is enough to show both $\mat{G}^T\cdot \mat{V}$ and $\mat{H}^T\cdot \mat{V}$ are i.i.d. Gaussian. In the following, we will show $\mat{G}^T\cdot \mat{V}$ are i.i.d. Gaussian and independent of $\mat{H}^T\cdot \mat{V}$, and the similar argument also holds for $\mat{H}^T\cdot \mat{V}$.

Let $\mat{V}_i = [\vec{v}_1,\dots,\vec{v}_i]$ and  $\mat{V}_i^{std}=[\vec{e}_1, \cdots, \vec{e}_i]$. Note that $\vec{v}_1,\dots,\vec{v}_q$ are unit vectors and orthogonal to each other. We first define orthogonal rotation matrices $\mat{R}_1, \mat{R}_2,\cdots$ recursively as follows. The matrix $\mat{R}_1$ will take $\vec{v}_1$ to $\vec{e}_1$. The matrix $\mat{R}_i$ will take $\vec{e}_j$ to $\vec{e}_j$ for any $j<i$ and takes $\mat{R}_{i-1}\cdots \mat{R}_1\vec{v}_i$ to $\vec{e}_i$. Note, $\mat{R}_i$ only depends on the first $i$ query vectors. We have $\mat{R}_{i}\cdots \mat{R}_1\mat{V}_i=\mat{V}_i^{std}$ for any $i\leq q$, and  $\mat{G}^T  \mat{V} = \mat{G}^T\cdot \mat{R}_1^{-1}\cdots \mat{R}_q^{-1}\cdot \mat{V}^{std}$. Define matrix $\mat{GR}_i \triangleq \mat{G}^T\cdot \mat{R}_1^{-1}\cdots \mat{R}_i^{-1}$. In the following, we use induction to show for any $i\leq q$, $\mat{GR}_i$ is i.i.d. Gaussian and independent of $\mat{H}^T\cdot \mat{V}$. It is enough to show that for any fixed $H$, for any $i\leq q$, $\mat{GR}_i$ is i.i.d. Gaussian.

For $i=1$, since $\mat{R}_1$ is determined by $\vec{v}_1$ which is independent of $\mat{G}^T$ and $\mat{R}_1$ is an orthogonal matrix, $\mat{GR}_1 = \mat{G}^T \mat{R}_1^{-1}$ is i.i.d. Gaussian.

Now, suppose $\mat{GR}_i$ is i.i.d. Gaussian, we will prove $\mat{GR}_{i+1} = \mat{GR}_i \cdot \mat{R}_{i+1}^{-1}$ is also i.i.d. Gaussian.  On one hand, $\mat{R}_{i+1}$ is determined by $\vec{v}_1, \cdots, \vec{v}_{i+1}$ which are determined by the response of the first $i$ queries, that is, determined by $\mat{M_2} \mat{V}_i$. We have
\begin{displaymath}
\begin{array}{ll}
 & \mat{M_2} \mat{V}_i  \\
= & \mat{U} \mat{G}^T \mat{R}_1^{-1}\cdots \mat{R}_i^{-1} \mat{V}^{std}_i + \frac{1}{\poly{n}} \mat{U}^{\bot} \mat{H}^T \mat{R}_1^{-1}\cdots \mat{R}_i^{-1} \mat{V}^{std}_i\\
= & (\mat{U}\times (\mat{G}^T \mat{R}_1^{-1}\cdots \mat{R}_i^{-1}) + \frac{1}{\poly{n}} \mat{U}^{\bot} \times (\mat{H}^T \mat{R}_1^{-1}\cdots \mat{R}_i^{-1}))\cdot \mat{V}^{std}_i
\end{array}
\end{displaymath}
It means $\mat{R}_{i+1}$ is determined by the first $i$ columns of matrix $\mat{GR}_i = \mat{G}^T \mat{R}_1^{-1}\cdots \mat{R}_i^{-1}$ and $\mat{H}^T \mat{R}_1^{-1}\cdots \mat{R}_i^{-1}$. Note by the inductive assumption, $\mat{GR}_i$ is i.i.d. Gaussian. Therefore, $\mat{R}_{i+1}$ is independent of the last $n-i$ columns of $\mat{GR}_i$.

On the other hand, $\mat{R}_{i+1}\vec{e}_j=\vec{e}_j$ for any $j\leq i$, and thus
$\mat{R}_{i+1}^{-1}=
  \begin{bmatrix}
    I_i & 0 \\
    0 & \mat{R}' \\
  \end{bmatrix}
$
where $I_i$ is the $i\times i$ identity matrix. Note the matrix $\mat{R}'$ is actually determined by the protocol, $\mat{U}$, $\mat{U}^{\bot}$, $\mat{H}$ and also the first $i$ columns of $\mat{GR}_i$, but it is independent of the last $n-i$ columns of $\mat{GR}_i$.  Consequently, in the multiplication of $\mat{GR}_i\times \mat{R}_{i+1}^{-1}$, the first $i$ columns are the same as those in $\mat{GR}_i$. For the other $n-i$ columns, the $a$-th element of $j$-th column is $\sum_{b\geq i+1} gr_{ab}r'_{b,j}$ where $gr_{ab}, r'_{b,j}$ are the elements in $\mat{GR}_i, \mat{R}'$ correspondingly. Since $r'_{b,j}$ is independent of the last $n-i$ columns of $\mat{GR}_i$, it is independent of $gr_{ab}$ when $b\geq i+1$. Since $\mat{GR}_i$ is i.i.d Gaussian and $\mat{R}'$ is an orthogonal matrix, the last $n-i$ columns of $\mat{GR}_{i+1}$ is also i.i.d. Gaussian and independent of the first $i$ columns.  Therefore, we show $\mat{GR}_{i+1} = \mat{G}^T\cdot \mat{R}_1^{-1}\cdots \mat{R}_{i+1}^{-1}$ is still i.i.d. Gaussian.

By induction $\mat{G}^T \mat{V}$ is i.i.d. Gaussian, and independent of $\mat{H}^T \mat{V}$. This finishes our proof.

\endproof

Obliviously, the same argument also holds for $\mat{M}_1$. Combining these results and Theorem~\ref{thm:rank}, together with Yao's minimax principle \cite{yao1977probabilistic},
\begin{theorem}
	Let constant $\eps>0$ be the error tolerance and let $\mat{M}$ be an $n\times n$ oracle matrix with adaptive queries.
	For integer $p<p'\leq n$, at least $p+1$ queries are necessary for any randomized algorithm to distinguish
  whether $\rank{\mat{M}}\le p$ or $\rank{\mat{M}} \ge p'$ with advantage $\ge \eps$.
\end{theorem}

\subsection{Lower Bound for Trace Estimation}
\label{sec:tra}

In this section, we lower bound the number of queries needed to approximate the trace $\tr{\mat{M}}$ of a matrix $\mat{M}$.
In particular we reduce this problem to triangle detection as will be proved in Theorem~\ref{thm:tri}. For the trace estimation problem, Avron and Toledo \cite{AT11} analyzed the convergence of randomized trace estimators via a similar matrix vector products framework. In their model, for an unknown matrix $\mat{M}$, they can access it via $\vec{v}^T\mat{M}\vec{v}$; while in our model, we only consider the right multiplication of the form $\mat{M}\vec{v}$.


\begin{theorem}\label{thm:trace}
For any integer $C > 0$ and symmetric $n \times n$ matrix $\mat{M}$ with entries in $\set{0, 1, 2, \ldots, n^3}$, the number of possibly adaptively chosen query vectors, with entries in $\set{0, 1, 2, \ldots, n^C}$, needed to approximate $\tr{\mat{M}}$ up to any relative error, is $\bigomega{n/\log n}$.
\end{theorem}
\begin{proof}
Suppose we had a possibly adaptive query algorithm making $q(n)$ queries which for a symmetric matrix $\mat{M}$, could approximate $\tr{\mat{M}}$ up to any relative error.
If $\mat{M} = \mat{A}^3$ for a symmetric matrix $\mat{A}$, we can run the trace esimation algorithm on $\mat{M}$ as follows: if $\vec{x}_1$ is the first query, we compute $\mat{A}\vec{x}_1$, then $\mat{A}(\mat{A}\vec{x}_1)$, then $\mat{A}(\mat{A}(\mat{A}\vec{x}_1)) = \mat{A}^3 \vec{x}_1$.
This then determines the second query $\vec{x}_2$, and we similarly compute $\mat{A}\vec{x}_2$, then $\mat{A}(\mat{A}\vec{x}_2)$, then $\mat{A}(\mat{A}(\mat{A}\vec{x}_2)) = \mat{A}^3\vec{x}_2$, etc. Thus, given only query access to $\mat{A}$, we can simulate the algorithm on $\mat{M} = \mat{A}^3$ with $3q(n)$ adaptive queries.

Now, it is well known that for an undirected graph $G$ with adjacency matrix $\mat{A}$, the trace $\tr{\mat{A}^3}/6$ is the number of triangles in $G$. By the argument above, it follows that with $3q(n)$ queries to $\mat{A}$, we can determine if $G$ has a triangle or has no triangles.
On the other hand, by Theorem~\ref{thm:tri} below, at least $\bigomega{n/\log n}$ queries to $\mat{A}$ are necessary for any adaptive algorithm to decide if there is a triangle in $G$.
Therefore $3q(n)=\bigomega{n/\log n}$ and hence we complete the proof with $q(n)=\bigomega{n/\log n}$.
\end{proof}

\subsection{Deciding if $\mat{M}$ is a Symmetric Matrix}\label{sec:sym}
\begin{theorem}\label{thm:symmetric}
Given an $n\times n$ matrix $\mat{M}$ over any finite field or over fields  $\R$ or $\C$, $O(\log(\frac{1}{\varepsilon}))$ queries are enough to test whether $\mat{M}$ is symmetric or not with probability $1-\varepsilon$.
\end{theorem}
\begin{proof}
We choose two random vectors $\vec{u}$ and $\vec{v}$, where over a finite field we choose from a uniform distribution and over fields $\R$ or $\C$ we choose the Gaussian distribution. We then compute $\mat{M}\vec{u}$ and $\mat{M}\vec{v}$. We declare $\mat{M}$ to be symmetric if and only if $\vec{u}^T\cdot\mat{M}\vec{v} = \vec{v}^T\cdot\mat{M} \vec{u}$. It is easy to check that if $\mat{M}$ is symmetric, the test will succeed. We then show if $\mat{M}$ is not symmetric, $\vec{u}^T\mat{M}\vec{v} \neq \vec{v}^T\mat{M} \vec{u}$ with constant probability, so we obtain success probability $1-\varepsilon$ by repeating the test $O(\log(\frac{1}{\varepsilon}))$ times.

Let $\mat{A}=\mat{M}-\mat{M}^T$. When $\mat{M}$ is not symmetric, $\mat{A}$ is not $0$. Thus, $\vec{u}^T\mat{M}\vec{v} = \vec{v}^T\mat{M} \vec{u}$ means $\vec{u}^T\mat{A}\vec{v} = 0$. We can treat this as a degree-$2$ polynomial in the entries of $\vec{v}^T$ and $\vec{u}$, i.e., this is $\sum_{i,j} u_i v_j \mat{A}_{i,j} = \sum_i u_i \sum_j v_j \mat{A}_{i,j}$. Thus, this is a non-zero polynomial and has at most constant probability of evaluating to $0$ for any underlying field. 
To see this, for each $i$, let $t_i = \sum_j v_j \mat{A}_{i,j}$. Then there will be at least one $t_i$ which is non-zero with probability at least $1/2$, for any underlying field. So now we get $\sum_i u_i t_i$. Fix all the $u_i$ except $u_i$ for a given $t_i$ that is non-zero. Then we obtain $S + u_i t_i$. Then if $u_i$ has at least two possible values, this is $0$ in one case and non-zero in the other case. So we obtain a probability of at least $1/4$ of detection overall.
\end{proof}

\subsection{Deciding if $\mat{M}$ is a Diagonal Matrix}\label{sec:diag}

Given an $n\times n$ matrix $\mat{M}$, we show that $\bigomega{\log\frac{1}{\eps}}$ queries are sufficient to test whether $\mat{M}$ is a diagonal matrix with error $\le \eps$.

The first query is an all ones vector which retrieves the sum of each row.
Then we take $\bigomega{\log\frac{1}{\eps}}$ random queries where each entry is uniformly sampled from $\zo$.
Every row containing non-zero entries off the diagonal can be detected with probability $1/2$ under such a random query, which implies bounded error $\le \eps$ after  $\bigomega{\log\frac{1}{\eps}}$ random queries.
Furthermore, this algorithm works over any field.

\subsection{Deciding if $\mat{M}$ is a Unitary Matrix}\label{sec:unitary}
Given an $n\times n$ complex matrix $\mat{M}$, we show $1$ query is enough to test whether $\mat{M}$ is unitary or not, that is $\mat{M}^*\mat{M}=\mat{M}\mat{M}^*=I$.

We choose a random Gaussian vector $\vec{v}$, and compute $\mat{M}\vec{v}$. We declare $\mat{M}$ to be unitary if and only if $|\mat{M}\vec{v}|_2 = |\vec{v}|_2$. It is easy to check that if $\mat{M}$ is unitary, the test will succeed. We then show if $\mat{M}$ is not unitary, $|\mat{M}\vec{v}|_2 \neq |\vec{v}|_2$ with probability $1$. Let the singular value decomposition of $\mat{M}$ be $\mat{M} = \mat{U} \Sigma \mat{V}^T$. We have $|\mat{M}\vec{v}|_2^2 = |\Sigma \vec{u}|_2^2$, where $\vec{u}=\mat{V}^T\vec{v}$ is a random Gaussian vector with $|\vec{u}|_2^2 = |\vec{v}|_2^2$. The diagonal values in $\Sigma$ are not all $1$ since $\mat{M}$ is not unitary. Consider $\sum_i \sigma_i^2 u_i^2$, where $\sigma_i = \Sigma_{i,i}$. We want this to equal $|\vec{v}|_2^2 = |\vec{u}|_2^2 = \sum_i u_i^2$, so this is $\sum_i u_i^2 (\sigma_i^2 - 1) = 0$. This is a non-zero polynomial and has probability $0$ of evaluating to $0$ since the $u_i^2$ are drawn from a continuous distribution.

\subsection{Approximating the Maximum Eigenvalue}\label{sec:maxeigen}
The upper bound is due to \cite{musco2015randomized}. Given a matrix $\mat{M}\in \R^{m\times n}$, we can $\eps$-approximate the maximum eigenvalue of $\mat{M}$ by taking a random vector $\vec{v}\in\R^n$ and computing $\mat{M}^r\vec{v}$ for $r=\bigo{\eps^{-0.5}{\log n}}$.
This requires $r$ adaptive oracle queries to $\mat{M}$.
See \cite{musco2015randomized} for details.
See \cite{Simchowitz2018} for a matching lower bound for adaptive queries.
A non-adaptive $\bigomega{n}$ lower bound is given in \cite{li2016tight}.


\section{Streaming and Statistics Problems}

In this section we discuss the following streaming and statistics problems: testing an all ones column/row and identical columns/rows; approximating row norms or finding heavy hitters; and computing the majority or parity of columns/rows.


\subsection{Testing Existence of an All Ones Column/Row}\label{sec:all1}

Given a matrix $\mat{M}\in \zo^{m\times n}$, we want to test if $\mat{M}$ has a column (or row) with all $1$ entries.
It is trivial to test whether $\mat{M}$ has an all $1$ column (or row) using $n$ queries, e.g. $\vec{e}_1,\dots, \vec{e}_n$.
We consider this problem both over $\mathbb{F}[2]$ and $\R$. Note in the case over $\R$, if we allow an arbitrary query vector, we can set one query $\vec{v} = \{1,2,4,8,...2^n\}$, and then reconstruct $\mat{M}$ exactly. Thus, in order to avoid such trivial cases, we also restrict the entries in the query to be in $\{0, 1, 2, \ldots, n^C\}$.

For testing the existence of an all ones column, we reduce the problem to the communication complexity of \disjoint.
\disjoint requires $\Omega(n)$ bits of communication to decide whether two sets with characteristic vectors $\vec{x},\vec{y}\in \zo^n$ are disjoint with constant probability, where the randomness is taken only over the coin tosses of the protocol (not over the inputs).
Suppose the fist $m-1$ rows in $\mat{M}$ each equal $\vec{x}^T$ while the last row equals $\vec{y}^T$. If we can decide whether $\mat{M}$ has an all ones column with $q$ non-adaptive queries $\vec{v}_1,\dots,\vec{v}_q$,
then we obtain a protocol for \disjoint with communication $q$ by letting Alice send a message $\left(\vec{x}^T\vec{v}_1,\dots, \vec{x}^T\vec{v}_q\right)$.
Thus from the communication complexity lower bound of \disjoint, $q = \Omega(n)$ queries over $\mathbb{F}[2]$ are necessary to test if there is an all ones column in $\mat{M}$, which shows that the na\"ive algorithm is already optimal.
For queries over $\R$, note that each entry $\vec{x}^T\vec{v}_j$ in the message is represented with ${\log n}$ bits,
and as a result $q\ge \bigomega{{n}/{\log n}}$.

Testing the existence of an all ones row with queries over $\R$ is trivial deterministically
by querying $\vec{v}=(1,1,\dots,1)$.
Next we study the query complexity of testing an all $1$s row deterministically with queries over $\mathbb{F}[2]$.
With any $q\le n-1$ queries $\mat{V}= \left[\vec{v}_1,\dots, \vec{v}_q\right]$, there is a non-zero vector $\vec{z}\ne\vec{0}$ such that $\vec{z}^T \mat{V}=\vec{0}$.
Therefore the query matrix $\mat{V}$ cannot distinguish whether a row is from $\vec{x}^T$ or $\vec{x}^T+\vec{z}^T$.
However, $\vec{x}^T$ and $\vec{x}^T+\vec{z}^T$ cannot be both all $1$ rows, and hence $n$ queries are necessary. This result also shows that the query complexity of the same problem over different fields might be quite different. We note for randomized algorithms, $O(\log(1/\epsilon))$ queries suffice over $\mathbb{F}[2]$ since the inner product of a row which is not all $1$s disagrees with the parity of the query with probability $1/2$.

Evaluating the OR/AND function of columns/rows of a Boolean matrix can be reduced to testing existence of an all 1 or all 0 column/row, and hence the same bounds follow.

\subsection{Identical Columns/Rows}\label{sec:identi}
\label{subsec:identical column/rows}

Given an $m\times n$ matrix $\mat{M}$, we want to test whether $\mat{M}$ has two identical columns or rows.
The trivial solution naively retrieves all information with $n$ queries (column vectors). In this section, we consider the query complexity over $\mathbb{F}[2]$.

Testing identical columns can be reduced to \disjoint.
Suppose Alice and Bob have $\vec{x},\vec{y}\in\set{0,1}^n$.
Let Alice expand her vector $\vec{x}$ to an $\frac{m}{2}\times n$ matrix $\mat{M}_1$ as follows:
the first row is $(1,\vec{x}^T)=(1,x_1,\dots,x_n)$; for $2\le i\le \frac{m}{2}$ the $i$-th row is $(1,z_1^{(i)},\dots, z_{n}^{(i)})$ where $z_j^{(i)}=1$ if $x_j=1$, and $z_j^{(i)}$ is uniformly random over $\set{0,1}$ if $x_j=0$, for $1\le j\le n$.
Bob expands his vector $\vec{y}$ to $\mat{M}_2$ similarly.
Putting $\mat{M}_1,\mat{M}_2$ together, we let $\mat{M}\eqdef \begin{bmatrix}
    \mat{M}_1 \\
    \mat{M}_2
\end{bmatrix}$.
Then $\mat{M}$ is an $m\times (n+1)$ matrix with the first column being all $1$s.
For $j\ge 2$, the $j$-th column is all $1$s if and only if $x_j=y_j=1$, in which case $\mat{M}$ has two identical rows of all $1$ entries.
For columns where $x_j,y_j$ are not both equal to $1$,
without loss of generality we may assume the $j$-th and $j'$-th columns satisfy $x_j=x_{j'}=0$ and $y_j=y_{j'}$.
Then two columns are identical only if $(z_j^{(2)},\dots,z_j^{(\frac{m}{2})})=(z_{j'}^{(2)},\dots,z_{j'}^{(\frac{m}{2})})$, which happens with probability $\le {1}/{2^{\frac{m}{2}-1}}$.
Therefore the overall probability of two not-all-ones columns in $\mat{M}$ being identical is bounded by ${n^2}/{2^{m/2}}$. Thus the error probability is less than $\epsilon$ if $m\geq 4\log(n/\epsilon)$.

That is, except for a small error $\epsilon$,
two identical columns in $\mat{M}$ are both all ones columns, which turns out to be equivalent to the case that two vectors $\vec{x},\vec{y}$ held by Alice and Bob are not disjoint.
Then, because \disjoint requires $\Omega(n)$ bits of communication, and after one oracle queries, Alice or Bob can communicate at most $m$ bits, at least $\Omega(n/m)$ oracle queries to $\mat{M}$ are necessary.

In the other hand, to test identical rows with error $\eps$, if suffices to make $q=\bigo{\log\left({m}/{\eps}\right)}$ random queries with each entry uniform random over $\zo$.
Since for every pair of distinct rows, a random query distinguishes them with probability $\frac{1}{2}$,
with $\left\lceil\log\left({m^2}/{\eps}\right)\right\rceil$ queries each pair of distinct rows is miscounted as identical with probability $\le {\eps}/{m^2}$.
By a union bound, the overall false-positive error is bounded by $\frac{\eps}{m^2}\cdot \binom{m}{2}  <\eps$, while there is no false-negative error since for all queries, identical rows always lead to identical outputs.

\subsection{Approximating Row Norms and Finding Heavy Hitters}\label{sec:norms}
To approximate the norms of each row in a matrix $\mat{M}\in \R^{m\times n}$, we recall the Johnson-Lindenstrauss lemma which guarantees that norms are roughly preserved when embedded to a lower dimensional space.
Thus, with $q=\bigo{\eps^{-2} \log m}$ and an $n\times q$ random query matrix $\mat{V}$, the output $\mat{M}\cdot\mat{V}$ preserves the row norms of $\mat{M}$ up to a $(1 \pm \eps)$-factor.

The above algorithm also gives a natural upper bound for finding heavy hitters in the matrix $\mat{M}$,
which requires finding all rows $\mat{M}_i$ with norm $|\mat{M}_i|_2^2\ge \frac{1}{10} |\mat{M}|_F^2$ and not outputting any row $\mat{M}_i$ with  $|\mat{M}_i|_2^2\le \frac{1}{20} |\mat{M}|_F^2$ (rows with norm in between the two quantities can be classified arbitrarily).
Again we use the Johnson-Lindenstrauss lemma to approximate all row norms and decide which row is a heavy hitter.

\subsection{Majority}\label{sec:majority}

Given a matrix $\mat{M}\in\zo^{m\times n}$, we want to compute the majority of rows or columns in $\mat{M}$.

The majority of each row in $\mat{M}$ is trivial with an all $1$ query and addition over $\R$.

For the majority of columns in $\mat{M}$,
we use a similar matrix $\mat{M}$ as that reduced from \disjoint in Section~\ref{subsec:identical column/rows} to obtain a lower bound.
More specifically, we consider $\vec{x},\vec{y}$ whose intersection is at most $1$.
Let $\mat{M}$ be obtained from $\vec{x},\vec{y}$ such that the first $m/2$ rows are identical to $\vec{x}$ and the remaining rows are identical to $y$. Thus, if $\mat{M}$ has a column with majority $1$, then the column must be all $1$s and we can conclude that $\vec{x}$ and $\vec{y}$ are not disjoint.
As a result, $\Omega({n/\log n})$ queries are necessary to compute the majority of columns in $\mat{M}$.

\subsection{Parity}\label{sec:parity}
For parity we consider a matrix $\mat{M}\in\zo^{m\times n}$ with only queries over $\F[2]$.
Computing the parity of rows in $\mat{M}$ is trivial by using a vector $(1,1,\dots,1)$.
However, to compute the parity of all columns of $\mat{M}$, we claim
at least $n$ queries are necessary.

To see this, let $\mat{V}$ be any $n\times q$ query matrix.
Note that the parity of columns of $\mat{M}$ remains the same if we sum up all the rows,
i.e., $\mat{M}'\eqdef \mat{P}\cdot \mat{M}$ has
the same parity on each column as $\mat{M}$, where $\mat{P}$ is defined to be
\[\mat{P} \eqdef \begin{bmatrix}
    1 & 1 & \dots & 1 \\
    0 & 0 & \dots & 0 \\
    \vdots & \vdots & \dots & \vdots \\
    0 & 0 & \dots & 0 \\
  \end{bmatrix}_{m\times m}
\]
Thus $\mat{M'}\mat{V}=\mat{PMV}$ is a $1\times q$ row vector followed by $m-1$ zero rows,
since $\mat{M'}$, as well as $\mat{P}$, is non-zero only in the first row.
Then we must have $q= \Omega(n)$ to obtain the output of $n$ parity instances from $\mat{M}' \mat{V}$.
Indeed, if we were to place the uniform distribution on $\mat{M}$ then its columns define
$n$ uniform parity bits, and for any fixed $\mat{V}$, we only obtain $q$ bits of information,
which is a contradiction to Yao's minimax principle (since there must be a fixed $V$
which succeeds with at least $2/3$ probability on this distribution).
This is a typical example illustrating the difference between left- and right-queries.


\section{Graph Problems}
In this section, we provide our results related to graph problems: testing graph connectivity in Section~\ref{sec:connect} and triangle detection in Section~\ref{sec:triangle}.

\subsection{Connectivity}\label{sec:connect}


\begin{theorem}
	Given the bipartite adjacency matrix $\mat{A} \in \{0,1\}^{n \times n}$ of a graph,
	we need $\bigomega{n/\log n}$ queries to decide whether the graph is connected with constant
        probability.
\end{theorem}
\begin{proof}
	Consider two row vectors $\vec{u},\vec{v}\in\zo^{n-1}$ and construct matrix $\mat{A}$ as follows. The first $n/2$ rows of $\mat{A}$ equal $\vec{u}$ and the rest are equal to $\vec{v}$. Also, add an all $1$s column to $\mat{A}$. Now, matrix $\mat{A}$ can be treated as a bipartite adjacency matrix of a graph $G$ with $n$ vertices in each part, where $\mat{A}_{i,j} = 1$ iff there is an edge from the $i$-th left vertex to the $j$-th right vertex. Since all left vertices connect with the $n$-th right vertex, the graph $G$ is disconnected if and only if there exists some right vertex which does not connect with any left vertices, that is, the corresponding column of matrix $\mat{A}$ is an all $0$s column. In another word, $G$ is disconnected if and only if the two vectors $\vec{u}$ and $\vec{v}$ are $0$ on the same position.

	Thus any algorithm that uses $q(n)$ non-adaptive queries on the right of $\mat{A}$ to decide the connectivity of $G$ immediately implies a protocol for set disjointness, provided we replace $1$s with $0$s in the input characteristic vectors to the set disjointness
problem. So the communication is at most $q(n)\log n$, thus $q(n)=\bigomega{n/\log n}$.	
\end{proof}

\begin{theorem}
	Given the signed edge-vertex incidence matrix $\mat{M}\in\set{0,\pm 1}^{n\times\binom{n}{2}}$ of a graph $G$ with $n$ vertices,
	the connectivity of $G$ can be decided with $\polylog{n}$ non-adaptive queries.
\end{theorem}

This follows from the main theorem of \cite{KLMMS17} (also proved in the work \cite{agm12}). By the following theorem, every cut of $G$ is multiplicatively approximated and hence $G$ is connected iff $H$ is connected, since a graph is disconnected iff it has a zero cut.

\begin{theorem}[\cite{KLMMS17}]

		There is a distribution on  $\binom{n}{2}\times \polylog{n}$ matrices $\mat{S}$ such that
		from $\mat{M}\mat{S}$, one can construct a $(1\pm 0.1)$-sparsifier ${H}$ of the graph $G$ with constant probability. Here, $\vec{x}^T \mat{L}_G x = (1 \pm 0.1) \vec{x}^T \mat{L}_H x$ for all $x$, with constant probability, where $\mat{L}_G$ and $\mat{L}_H$ are the corresponding graph Laplacians.
	\end{theorem}

	By the above, every cut of $G$ is multiplicatively approximated and hence $G$ is connected iff $H$ is connected, since a graph is disconnected iff it has a zero cut.

\subsection{Triangle Detection}\label{sec:triangle}
\label{sec:tri}


\begin{theorem}\label{thm:tri}
If an $n \times n$ matrix $\mat{A}$ is the adjacency matrix of a graph $G$, then determining whether $G$ contains a triangle or not requires $\bigomega{n/\log n}$ queries, even for randomized algorithms succeeding with constant probability.
\end{theorem}
\begin{proof}
To obtain a lower bound on $q(n)$, we use a $2$-player communication lower bound of counting the number of triangles in a graph $G$, where the edges are distributed across the two players, Alice and Bob.
Namely, it is known \cite{BKS02, FKMSZ08, HSSW12} that if Alice has a subset of the edges of $G$, and Bob has the remaining (disjoint) subset of edges of $G$, then the multiround randomized communication complexity of deciding if there is a triangle in $G$ is $\Omega(n^2)$.
Alice can view her subset of edges as an adjacency matrix $\mat{A}'$, and Bob can view his subset of edges as an adjacency matrix $\mat{A}''$, so that $\mat{A} = \mat{A}' + \mat{A}''$. To execute the query algorithm on $\mat{A}$, Alice sends $\mat{A}'\vec{x}_1$ to Bob, who computes $\mat{A}''\vec{x}_1$ followed by $\mat{A}'\vec{x}_1 + \mat{A}''\vec{x}_1 = \mat{A}\vec{x}_1$, and sends the result back to Alice.
Alice then possibly adaptively chooses $\vec{x}_2$, which is also known to Bob who knows $\vec{x}_1$ and $\mat{A}\vec{x}_1$, and sends Bob $\mat{A}'\vec{x}_2$, from which Bob can compute $\mat{A}''\vec{x}_2$ and $\mat{A}\vec{x}_2 = \mat{A}'\vec{x}_2 + \mat{A}''\vec{x}_2$.
This process repeats until the entire $q(n)$ queries have been executed, at which point Bob, by the success guarantee of the algorithm, can decide if $G$ contains a triangle with say, probability at least $2/3$.
Because of the bounds on the bit complexity of the queries while the total communication is $\bigo{q(n) n \log n}$, which must be $\Omega(n^2)$, and consequently $q(n) = \Omega(n/\log n)$, as desired.
\end{proof}




\section{Conclusions}
We initiated the study of querying a matrix through matrix-vector products. 
We illustrated that for some quantities, if one can only query matrix-vector products on one side, the problem becomes harder. We also illustrated the importance of the underlying field defining the matrix-vector products, as well as the representation of the graph for graph problems. Given connections to sketching algorithms, streaming, and compressed sensing, we believe this area deserves its own study. Some interesting open questions are for computing matrix norms, such as Schatten-$p$ norms, for which tight bounds in streaming and communication complexity models remain elusive; for recent work on this see \cite{LW16,LW17,BCKLWY18}. Given the success of our model in proving lower bounds for approximate rank, which we also do not have streaming or communication lower bounds for, perhaps tight bounds in our query model are possible for matrix norms. Such bounds may give insight for other models. 



\bibliography{ref}
\newpage
\appendix

\end{document}